\documentclass[english,3p,twocolumn,superscriptaddress,pra]{revtex4-1}
\usepackage[toc,page]{appendix}
\usepackage[latin9]{inputenc}
\usepackage{amsmath}
\usepackage{amssymb, enumerate}
\usepackage{color}
\usepackage{babel}
\usepackage{graphicx}
\usepackage{epstopdf}
\usepackage{float}
\usepackage{thmtools}
\usepackage{thm-restate}
\usepackage{hyperref}
\usepackage{cleveref}
\usepackage{enumerate}
\usepackage{lipsum}
\usepackage{fixmath}
\usepackage{dsfont}
\usepackage{tikz}
\usetikzlibrary{arrows.meta,bending,matrix,positioning}
\makeatletter

\@ifundefined{textcolor}{}
{%
 \definecolor{BLACK}{gray}{0}
 \definecolor{WHITE}{gray}{1}
 \definecolor{RED}{rgb}{1,0,0}
 \definecolor{GREEN}{rgb}{0,1,0}
 \definecolor{BLUE}{rgb}{0,0,1}
 \definecolor{CYAN}{cmyk}{1,0,0,0}
 \definecolor{MAGENTA}{cmyk}{0,1,0,0}
 \definecolor{YELLOW}{cmyk}{0,0,1,0}
}

\newtheorem{theorem}{Theorem}
\newtheorem{corollary}[theorem]{Corollary}
\newtheorem{proposition}[theorem]{Proposition}
\newtheorem{example}[theorem]{Example}
\newtheorem{remark}[theorem]{Remark}
\newtheorem{lem}[theorem]{Lemma}
\newtheorem{definition}[theorem]{Definition}

\newenvironment{proof}[1][Proof]{\noindent\textit{#1.} }{\ \rule{0.5em}{0.5em}}
\newenvironment{sproof}[1][Sketch of a proof]{\noindent\textit{#1.} }{\ \rule{0.5em}{0.5em}}

\makeatother

\begin{document}

\title{On the edge of the set of no-signaling assemblages}
\author{Micha{\l} Banacki}
\affiliation{International Centre for Theory of Quantum Technologies, University of Gda\'{n}sk, Wita Stwosza 63, 80-308 Gda\'{n}sk, Poland}
\affiliation{Institute of Theoretical Physics and Astrophysics, Faculty of Mathematics, Physics and Informatics, University of Gda\'{n}sk, Wita Stwosza 57, 80-308 Gda\'{n}sk, Poland}
\author{Ricard Ravell Rodr{\'i}guez}
\affiliation{International Centre for Theory of Quantum Technologies, University of Gda\'{n}sk, Wita Stwosza 63, 80-308 Gda\'{n}sk, Poland}
\author{Pawe{\l} Horodecki}
\affiliation{International Centre for Theory of Quantum Technologies, University of Gda\'{n}sk, Wita Stwosza 63, 80-308 Gda\'{n}sk, Poland}
\affiliation{Faculty of Applied Physics and Mathematics, National Quantum Information Centre, Gda\'{n}sk University of Technology, Gabriela Narutowicza 11/12, 80-233 Gda\'{n}sk, Poland} 
%%%%%%%%%%%%%%%%%%%%%

\begin{abstract}

Following recent advancements, we consider a scenario of multipartite postquantum steering and general no-signaling assemblages. We introduce the notion of the edge of the set of no-signaling assemblages and we present its characterization. Next, we use this concept to construct witnesses for no-signaling assemblages without an LHS model. Finally, in the simplest nontrivial case of steering with two untrusted subsystems, we discuss the possibility of quantum realization of assemblages on the edge. In particular, for three-qubit states, we obtain a no-go type result, which states that it is impossible to produce assemblage on the edge using measurements described by POVMs as long as the rank of a given state is greater than or equal to 3.
\end{abstract}

%%%%%%%%%%%%%%%%%%%%%


\maketitle

\section{Introduction}

Quantum theory provides us with phenomena going beyond any classical description or intuition. The most striking example of this statement is the possibility of obtaining correlations that cannot be explained by any local and realistic theory \cite{EPR1935, Bell}. Another non-classical phenomenon of quantum mechanics is encapsulated in the idea of quantum steering, proposed by von Neumann \cite{S36} and reintroduced in a modern formulation in \cite{WJD07}. The mathematical description of multipartite steering is given by the notion of assemblage consisting of subnormalized states describing subsystem of a chosen party, indexed according to measurements performed by other parties, and fulfilling a set of no-signaling constraints. Forgetting about the quantum nature of this collection of subnormalized states, one can define a notion of an abstract no-signaling assemblage based only on no-signaling constraints \cite{SBCSV15}. As it is known that not all no-signaling assemblages admit quantum realization (they do not come from a steering scenario described by the rules of quantum mechanics), there is room for the idea of postquantum steering. While our knowledge of the set of no-signaling assemblages has been vastly expanded \cite{SAPHS18,HS18}, the structural relations between various convex subsets of this set (set of assembles which admit quantum realization or an LHS model) are not fully understood.

This paper aims to discuss this structural relationship, following the general idea of the edge state \cite{BCH}, defined concerning a given convex subset of a considered set of quantum states. We introduce a similar concept in the setting of postquantum steering. We characterize its structural properties and we provide some applications.

In Section \ref{sII} we recall the notion of a multipartite quantum steering and a generalized concept of no-signaling assemblage. Moreover, in Section \ref{sIII}, we introduce a definition of the edge of no-signaling assemblages and present its characterization. In Section \ref{sIV} we formulate a notion of witnesses for no-signaling assemblages which do not admit an LHS model. In particular, we provide construction of witnesses starting from particular edge assemblages. Section \ref{sV} is dedicated to the problem of quantum realization of assemblages on the edge of no-signaling, in the simplest tripartite steering scenario with two uncharacteristic subsystems (each with binary settings and outcomes). Finally, Section \ref{sVI} consists of short discussion.

\section{No-signaling assemblages}\label{sII}

Nowadays the idea of quantum steering \cite{S36,WJD07} has become more relevant \cite{CS17, UCNG20} -  for example as a tool for certification of entanglement in a setting with untrusted parties. Recently, a typical bipartite scenario of quantum steering has been generalized to accommodate many uncharacterized (untrusted) systems \cite{CS15} and introduce assemblages even beyond quantum description \cite{SBCSV15, SAPHS18, HS18, SHSA20}.

Consider a (n+1)-partite steering scenario in which n distant untrusted subsystems $A_i$ share a quantum state $\rho$ with the distant trusted subsystem $B$ (described by $d$ dimensional Hilbert space) and subnormalized states of subsystem B conditioned upon uncharacterized measurements on subsystems $A_i$ are given by 
\begin{equation}\label{assemblage}
\sigma_{\textbf{a}|\textbf{x}}=\mathrm{Tr}_{A_1,\ldots, A_n}(M^{(1)}_{a_1|x_1}\otimes\ldots \otimes M^{(n)}_{a_n|x_n}\otimes \mathds{1}\rho)
\end{equation}where $\textbf{a}|\textbf{x}=a_1\ldots a_n|x_1\ldots x_n$ and any $M^{(i)}_{a_i|x_i}$ is an element of POVM (positive operator value measure) corresponding to the measurement outcome $a_i\in \left\{0,\ldots, \mathcal{A}_i-1\right\} $ of the measurement setting $x_i\in \left\{0,\ldots, \mathcal{X}_i-1\right\}$ related to subsystem $A_i$. We define \textit{quantum assemblage} as any collection of subnormalized states on $d$ dimensional space $\Sigma=\left\{\sigma_{\textbf{a}|\textbf{x}}\right\}_{\textbf{a},\textbf{x}}$, which can be obtained according to expression (\ref{assemblage}). Note that $\Sigma$ fulfills a set of no-signaling constraints, which may be used to formulate an abstract notion of \textit{no-signaling assemblage} by forgetting about origin of $\Sigma$ based on measurements performed on a quantum state.

\begin{definition}\label{NS-def} A \textit{no-signaling assemblage} $\Sigma=\left\{\sigma_{\textbf{a}|\textbf{x}}\right\}_{\textbf{a},\textbf{x}}$ is a collection of subnormalized states $\sigma_{\textbf{a}|\textbf{x}}$ acting on $d$ dimensional Hilbert space, for which
\begin{equation}\label{def11}
\forall_{x_1,\ldots, x_n} \sum_{a_1,\ldots, a_n} \sigma_{\textbf{a}|\textbf{x}}=\rho_B,
\end{equation}where $\rho_B$ is a state, and for any possible subset of indexes $I=\left\{i_1,\ldots, i_s\right\}$ with $1\leq s<n$
\begin{equation}\label{def12}
\forall_{x_1,\ldots, x_n} \sum_{a_j,j\notin I} \sigma_{\textbf{a}|\textbf{x}}=\sigma_{a_{i_1}\ldots a_{i_s}|x_{i_1}\ldots x_{i_s}}.
\end{equation}
\end{definition}It has been already shown \cite{G89,HJW93}, that any such assemblage for $n=1$ admits a quantum realization like in (\ref{assemblage}). Nevertheless, for $n>1$ this is no longer the case and the set of quantum assemblages is then a nontrivial subset of the set of all no-signaling assemblages abstractly defined as in Definition \ref{NS-def}. 

Inside this convex subset of quantum assemblages one can single out another nontrivial convex subset representing the steering scenarios with classically correlated systems \cite{SAPHS18} - the set of LHS (local hidden state) assemblages. We say that a no-signaling assemblage admits an \textit{LHS model} if it can be represented by
\begin{equation}\label{LHS}
\sigma_{\textbf{a}|\textbf{x}}=\sum_j q_j \prod_i^np^{(A_i)}_j(a_i|x_i)\rho_j
\end{equation}where $q_j\geq 0, \sum_j q_j=1$, $\rho_j$ are some states of characterized subsystem B and $p^{(A_i)}_i(a_i|x_i)$ denotes conditional probability distributions for each of uncharacterized subsystems $A_i$ respectively. Note that any LHS assemblage (\ref{LHS}) can be also described by $\sigma_{\textbf{a}|\textbf{x}}=\sum_j q_j p_j(\textbf{a}|\textbf{x})\rho_j$ where each $L_j=\left\{p_j(\textbf{a}|\textbf{x})\right\}_{\textbf{a},\textbf{x}}$ denotes a deterministic conditional probability distribution (i.e. extremal point in the polytope of local correlations). Indeed, this follows from the fact that for any $i$ all $p^{(A_i)}_i(a_i|x_i)$ can be realized as convex combinations of deterministic distributions (on a single party).

From now on we will omit subscript notation and we will write no-signaling assemblage simply as $\Sigma=\left\{\sigma_{\textbf{a}|\textbf{x}}\right\}$. For further convenience, for any no-signaling box $P=\left\{p(\textbf{a}|\textbf{x})\right\}$, we define set of indexes $I_P=\left\{\textbf{a}|\textbf{x}:p(\textbf{a}|\textbf{x})\neq 0\right\}$. Note that if $\textbf{a}|\textbf{x}=a_1\ldots a_n|x_1\ldots x_n$ where $a_i\in \left\{0,\ldots, \mathcal{A}_i-1\right\}$ and $x_i\in \left\{0,\ldots, \mathcal{X}_i-1\right\}$, then the cardinality of $I_L$ is equal to $|I_L|=\prod_{i}^n \mathcal{X}_i$, when $L$ denotes some local deterministic box ($I_L$ consists of indexes of positions in $L$ occupied by probabilities equal to $1$). Finally, let $R_{\textbf{a}|\textbf{x}}$ stands for projection on the image of $\sigma_{\textbf{a}|\textbf{x}}$, i.e. $\mathrm{Im}(\sigma_{\textbf{a}|\textbf{x}})$ .

\section{Edge of the set of no-signaling assemblages}\label{sIII}

Let $S_2$ denotes the set of all no-signaling assemblages $\Sigma=\left\{\sigma_{\textbf{a}|\textbf{x}}\right\}$ related to the scenario with fixed number of untrusted parties and with fixed numbers of labels from settings and outcomes. Let $S_1$ denotes its subset consisting of all assemblages which are admitting an LHS model. Note that both sets are convex and compact. Following the general spirit considered in \cite{BCH}, we introduce the following definition.

\begin{definition}\label{def_edge}
We say that $\Sigma\in S_2$ is on the edge of the set of no-signaling assemblages if for any convex decomposition $\Sigma=\epsilon\Sigma_1+(1-\epsilon)\Sigma_2$ where $\Sigma_1\in S_1$, $\Sigma_2\in S_2$, we have $\epsilon=0$.
\end{definition}
\begin{figure}[H]
\includegraphics[width=0.45\textwidth]{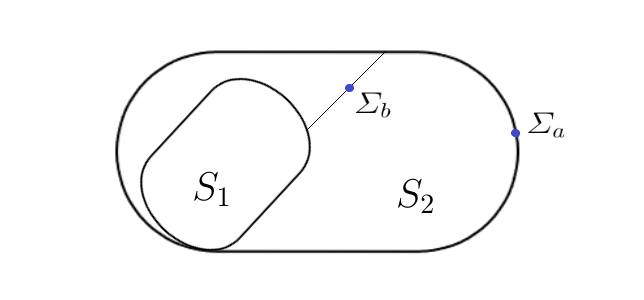}
\caption{Schematic representation of the edge - $\Sigma_a$ is on the edge and $\Sigma_b$ is not.}
\end{figure}
Observe that any extremal no-signaling assemblage is either LHS or it is on the edge, but not all edge assemblages are extremal (see example \ref{ex}).
Note that assemblage $\Sigma$ is not on the edge if and only if there exists $\epsilon>0$, some pure state $|\psi\rangle \langle \psi|$ and a deterministic box $L=\left\{ p(\textbf{a}|\textbf{x})\right\}$ such that all operators $\tilde{\sigma}_{\textbf{a}|\textbf{x}}=\sigma_{\textbf{a}|\textbf{x}}-\epsilon p(\textbf{a}|\textbf{x})|\psi\rangle \langle \psi|$ are positive. If this is the case, we will say that it is possible to subtract an LHS part from a given assemblage $\Sigma$. It is so because collection $\left\{\tilde{\sigma}_{\textbf{a}|\textbf{x}}\right\}$ by definition fulfills no-signaling conditions, the only obstruction for subtraction of an LHS part is the fact that operator related to the given position $\textbf{a}|\textbf{x}$ will no longer be positive.

Note that assemblages on the edge are precisely that which admit the maximal value of so-called steerable (steering) weight \cite{SNC04,CS17}. Therefore, the membership question regarding the edge of the set of no-signaling assemblages can be stated in the form of the SDP problem. Here we present an analytical procedure that determines an answer for the evoked question (without the SDP approach). To do this, we invoke the following lemmas.

\begin{lem}\label{lemma1}
Consider a positive operator $\sigma\in M_d(\mathbb{C})_{+}$ and a pure state $|\psi\rangle\langle\psi|\in M_d(\mathbb{C})_{+}$, then $\sigma - \epsilon |\psi\rangle\langle\psi| \geq 0$ for some $\epsilon>0$ if and only if $|\psi\rangle \in \mathrm{Im}(\sigma)$.
\end{lem}

By Lemma \ref{lemma1} and previous discussion it is obvious that one can subtract an LHS part from a given assemblage $\Sigma$ if and only if there exists a local deterministic box $L$ and a normalized vector $|\psi\rangle$ such that $|\psi\rangle \in \bigcap_{\textbf{a}|\textbf{x}\in I_L}\mathrm{Im}(\sigma_{\textbf{a}|\textbf{x}})$.
\begin{remark}\label{rem}
 Note that because of that, in the particular
 case of $d=2$, one can subtract an LHS part from assemblage $\Sigma$ if and only if there exists $L$ such that $\left\{\sigma_{\textbf{a}|\textbf{x}}: \textbf{a}|\textbf{x}\in I_L\right\}$ does not contain $0$ or two elements proportional to different rank one states.
\end{remark}

The above observation may be reformulated as the next lemma.

\begin{lem}\label{protocol}
An LHS part can be subtracted from assemblage $\Sigma$ if and only there exists a local deterministic box $L$ such that $\mathrm{det}\left(\prod_{\textbf{a}|\textbf{x}\in I_L}R_{\textbf{a}|\textbf{x}} - \mathds{1}\right)= 0$.
\end{lem}
\begin{proof} Fix given local deterministic box $L$ and enumerate elements of the set $\left\{\sigma_{\textbf{a}|\textbf{x}}: \textbf{a}|\textbf{x}\in I_L\right\}$ as $\left\{\sigma_i:i=1,\ldots ,|I_L|\right\}$. Accordingly to that enumeration let us put $R_i$ instead $R_{\textbf{a}|\textbf{x}}$.

Assume that an LHS part (related to $L$) can be subtracted from a given assemblage. Then there exists normalized vector $|\psi\rangle\in \bigcap_{i}^{|I_L|}\mathrm{Im}(\sigma_i)$. From this we have $R_i|\psi\rangle=|\psi\rangle$ for any $i$ and as a consequence $R_{1}R_{2}\ldots R_{|I_L|} |\psi\rangle = |\psi\rangle$, so $\mathrm{det}\left(\prod_i^{|I_L|}R_{i} - \mathds{1}\right)= 0$ since $|\psi\rangle\neq 0$.

For the opposite implication assume that $\mathrm{det}\left(\prod_{i}^{|I_L|}R_i - \mathds{1}\right)=0$. If so then there exists a normalized vector $|\psi \rangle$ such that $R_{1}R_{2}...R_{|I_L|}|\psi\rangle = |\psi\rangle$. Define $|\phi\rangle =R_{2}...R_{|I_L|}|\psi\rangle$.
Due to that we have $|||\phi\rangle|| \leq |||\psi\rangle||$ and $R_{1} |\phi\rangle = |\psi\rangle$. This implies the following decomposition $|\phi\rangle= |\psi\rangle + |\psi\rangle_{\perp}$ with $|\psi\rangle \perp |\psi\rangle_{\perp}$ and opposite inequality $|||\phi\rangle|| \geq |||\psi\rangle||$. Finally, $|\phi\rangle = |\psi\rangle$ and $R_{1}|\psi\rangle = |\psi\rangle$. Iteration of this procedure provides that $|\psi\rangle$ is an eigenvector (corresponding to eigenvalue $1$) of $R_i$ for each $i=1, \ldots, |I_L|$.
\end{proof}

This leads to the final theorem.

\begin{theorem}\label{char}
An assemblage $\Sigma$ is on the edge of the set of no-signaling assemblages if and only if $\mathrm{det}\left(\prod_{L\in \mathcal{L}}\left(\prod_{\textbf{a}|\textbf{x}\in I_L}R_{\textbf{a}|\textbf{x}} - \mathds{1}\right)\right)\neq 0$, where $\mathcal{L}$ denotes a set of all local deterministic boxes.
\end{theorem}

There is yet another simple way to check that a given assemblage is not on the edge.

\begin{lem}\label{rank} 
Let $\Sigma$ be a no-signaling assemblage of operators $\sigma_{\textbf{a}|\textbf{x}}$ acting on $d$-dimensional space. Then it is not on the edge if there exists a local deterministic box $L$ such that $\sum_{\textbf{a}|\textbf{x}\in I_L}\mathrm{rank}(\sigma_{\textbf{a}|\textbf{x}})>(|I_L|-1)d$.
\end{lem}

\begin{proof}Enumerate all elements by $\left\{\sigma_i:i=1,\ldots,|I_L|\right\}$. Because $\sum_i^{|I_L|}\mathrm{rank}(\sigma_i)>(|I_L|-1)d$ we have $\mathrm{rank}(\sigma_1)+\mathrm{rank}(\sigma_2)>d$ and therefore $V_1=\mathrm{Im}(\sigma_1)\cap \mathrm{Im}(\sigma_2)\neq \left\{0\right\}$ with $\mathrm{dim}\ V_1\geq \mathrm{rank}(\sigma_1)+\mathrm{rank}( \sigma_2)-d$. For similar reason we also have $\mathrm{dim}\ V_1+\mathrm{rank}(\sigma_3)>d$ (as $\mathrm{rank}(\sigma_1)+\mathrm{rank}(\sigma_2)+\mathrm{rank}(\sigma_3)>2d$) and $V_2=V_1\cap \mathrm{Im}(\sigma_3)\neq \left\{0\right\}$. Iterating this argument, we derive at $V_{|I_L|-1}=\bigcap_i^{|I_L|}\mathrm{Im}(\sigma_i)\neq \left\{0\right\}$, so there exists a non-zero vector which belongs to the images of all operators $\sigma_i$.
\end{proof}

As a simple consequence of the previous lemma, we obtain the following corollary.

\begin{corollary}\label{corollary}Let $\Sigma$ be a no-signaling assemblage of operators $\sigma_{\textbf{a}|\textbf{x}}$ acting on $d$-dimensional space. If $\Sigma$ is on the edge, then $\sum_{\textbf{a}|\textbf{x}}\mathrm{rank}(\sigma_{\textbf{a}|\textbf{x}})\leq \left(\prod_{i}^n \mathcal{X}_i-1\right)\left(\prod_{i}^n \mathcal{A}_i\right)d$.
\end{corollary}
\begin{proof}Consider the set $\Lambda$ of all positions $\textbf{a}|\textbf{x}$ in the assemblage $\Sigma$. Note that $|\Lambda|=\prod_{i}^n \mathcal{X}_i\prod_{i}^n \mathcal{A}_i$. Now observe that one can define $\prod_{i}^n \mathcal{A}_i$ disjoint subsets $\Lambda_i$ of $\Lambda$ such that $\Lambda=\cup_i \Lambda_i$ and $\Lambda_i=I_{L_i}$ for some local deterministic box $L_i$. By the Lemma \ref{rank} we see that if $\Sigma$ is on the edge, sum of ranks of elements $\sigma_{\textbf{a}|\textbf{x}}$ related to a given $\Lambda_i$ cannot exceed $\left(|I_{L_i}|-1\right)d=\left(\prod_{i}^n \mathcal{X}_i-1\right)d$.
\end{proof}

Let us consider a particular steering scenario with $n=2$ and $a,b,x,y\in \left\{0,1\right\}$, it is convenient to see any such no-signaling assemblage $\Sigma=\left\{\sigma_{ab|xy}\right\}$ as the following box 
\begin{equation}\label{box}
\Sigma=\begin{pmatrix}
\begin{array}{cc|cc}
 \sigma_{00|00} &  \sigma_{01|00} & \sigma_{00|01} &  \sigma_{01|01} \\  
 \sigma_{10|00} & \sigma_{11|00}& \sigma_{10|01}  & \sigma_{11|01} \\ \hline
 \sigma_{00|10} & \sigma_{01|10} & \sigma_{00|11}&  \sigma_{01|11}  \\
   \sigma_{10|10} & \sigma_{11|10} & \sigma_{10|11} & \sigma_{11|11}
\end{array}
\end{pmatrix}.
\end{equation}No-signaling conditions (\ref{def12}) have now simple interpretation. Namely, in each row sum of two operators on the right-hand side must be equal to the sum of two operators on the left-hand side (i.e. $\sum_{b}\sigma_{ab|xy}=\sum_b\sigma_{ab|xy'}$), while similarly in each column the sum of two operators in the upper part must be equal to the sum of two operators in the lower part (i.e. $\sum_{a}\sigma_{ab|xy}=\sum_a\sigma_{ab|x'y}$). Normalization condition (\ref{def11}) is encoded into the fact that $\mathrm{Tr}(\sum_{a,b}\sigma_{ab|xy})=1$ for all pairs $x,y$.
Note that in considering setting (i.e. $a,b,x,y\in \left\{0,1\right\}$) there are $16$ deterministic boxes $L_{\alpha\beta\gamma\delta}=\left\{p_{\alpha\beta\gamma\delta}(ab|xy)\right\}$, defined by conditional probabilities ($\oplus$ stands here for addition modulo $2$ and $\alpha,\beta, \gamma, \delta\in\left\{0,1\right\}$).
\begin{equation}
 p_{\alpha\beta\gamma\delta}(ab|xy)=
\begin{cases}
1\ \ \ \mathrm{for}\ a=\alpha x\oplus \beta, b=\gamma y\oplus \delta \\
0 \ \ \  \mathrm{otherwise}.
\end{cases}
\end{equation}Using graphical presentation like in (\ref{box}), any $L_{\alpha\beta\gamma\delta}$ can be seen as a box with four positions occupied by $1$ forming a rectangle (and other positions occupied by $0$). Therefore, one can subtract an LHS part from $\Sigma$ if and only if there exist a rectangle (one of $16$ possible) with vertexes given by positions occupied by operators with common vector in their images, as it was stated in Lemma \ref{lemma1} (see for example a graphical presentation of possible rectangle given by color red in (\ref{rectangle})).
\begin{equation}\label{rectangle}
\Sigma=\begin{pmatrix}
\begin{array}{cc|cc}
 \color{red} \sigma_{00|00}\color{black} & \ldots  & \color{red}\sigma_{00|01}\color{black} &  \ldots \\  
  \ldots & \ldots& \ldots  & \ldots\\ \hline
 \color{red} \sigma_{00|10}\color{black} & \ldots &  \color{red}\sigma_{00|11}\color{black} &  \ldots \\
 \ldots  & \ldots & \ldots & \ldots
\end{array}
\end{pmatrix}.
\end{equation}In particular if $d=2$ according to Remark \ref{rem} this can be done if and only if among operators from rectangle there is no $0$ nor two operators proportional to different rank one states - we will use graphical interpretation (\ref{box}) extensively.

\section{Construction of witnesses}\label{sIV}

Consider a no-signaling assemblage $\Sigma=\left\{\sigma_{\textbf{a}|\textbf{x}}\right\}$ with $\sigma_{\textbf{a}|\textbf{x}}\in M_d(\mathbb{C})_{+}\subset M_d(\mathbb{C})_{sa}$ where $M_d(\mathbb{C})_{sa}$ stands for a space of hermitian operators. Then $\Sigma$ can be seen as an element in a real Hilbert space $\bigoplus_{\textbf{a}|\textbf{x}}M_d(\mathbb{C})_{sa}$. Using theorem of Riesz and Hahn-Banach type of reasoning \cite{KR} for sets $S_1\subset S_2$, one can introduce a notion of a witness for no-signaling assemblages beyond LHS description. We say that $W\in \bigoplus_{\textbf{a}|\textbf{x}}M_d(\mathbb{C})_{sa}$ is a witness if $\mathrm{Tr}(W\Sigma_1)\geq 0$ for all $\Sigma_1\in S_1$ and there exists $\Sigma_2\in S_2\setminus S_1$ such that $\mathrm{Tr}(W\Sigma_2)< 0$. Note that $\Sigma\in S_2\setminus S_1$ if and only if there exist a witness $W$ such that $\mathrm{Tr}(W\Sigma)< 0$. In fact any assemblage belongs to the nontrivial subspace $V\subset \bigoplus_{\textbf{a}|\textbf{x}}M_d(\mathbb{C})_{sa}$ defined by no-signaling conditions, so one may consider only witnesses from $V$ - however we will not restrict our attention to $V$, as the following construction of witnesses out of edge assemblages may lead to witnesses beyond $V$ (discussion on optimality \cite{BCH22} of witnesses will be considered elsewhere).

Let us introduce a set $\mathcal{Z}$ which consists of all $Z\in \bigoplus_{\textbf{a}|\textbf{x}}M_d(\mathbb{C})_{sa}$ such that $\mathrm{Tr}(Z\Sigma)\geq 0$ for any no-signaling assemblage $\Sigma$. Any witness can be expressed in a canonical form related to particular assemblage on the edge (compare with Theorem 4 in \cite{BCH}).
\begin{proposition}
Let $W$ be a witness, then $W=Z-\epsilon \mathds{1}$ where $Z\in \mathcal{Z}$ such that $\mathrm{Tr}(Z\Sigma)=0$ for some edge assemblage $\Sigma$ and $\epsilon>0$.
\end{proposition}

Note that $\Sigma$ is on the edge if and only if for any local deterministic box $L$, there is no pure state $|\psi\rangle\in \mathbb{C}^d$ such that $|\psi\rangle \in \bigcap_{\textbf{a}|\textbf{x}\in I_L} \mathrm{Im}(R_{\textbf{a}|\textbf{x}})$. Let $\Sigma$ be on the edge. Define $Z\in \bigoplus_{\textbf{a}|\textbf{x}}M_d(\mathbb{C})_{sa}$ by $Z_{\textbf{a}|\textbf{x}}=\mathds{1}_d-R_{\textbf{a}|\textbf{x}}$. Obviously $\mathrm{Tr}(Z\Sigma)=0$ and there exists $\epsilon$ such that $0<\epsilon\leq \inf_{LHS} \mathrm{Tr}(Z\Sigma_{LHS})$. Indeed, for any LHS assemblage of the form $L\otimes |\psi\rangle\langle \psi|$, there exists $(\textbf{a}|\textbf{x})\in I_L$ such that $\mathrm{Tr}(|\psi\rangle\langle \psi|Z_{\textbf{a}|\textbf{x}})\neq 0$, because if not, then $|\psi\rangle \in \bigcap_{\textbf{a}|\textbf{x}\in I_L} \mathrm{Im}(R_{\textbf{a}|\textbf{x}})$ which is in contradiction with the fact that $\Sigma$ in on the edge.
If so, we may define a witness detecting $\Sigma$ as $W=Z-\frac{\epsilon}{\mathrm{Tr}(\Sigma)}\mathds{1}_{\oplus}$ where $\mathds{1}_{\oplus}=\bigoplus_{\textbf{a}|\textbf{x}}\mathds{1}_d$ and $\mathrm{Tr}(\Sigma)$ does not depend on particular $\Sigma$. Note that by putting maximal $\epsilon$ we obtain $W$ such that there exists a $\Sigma_{LHS}$ for which $\mathrm{Tr}(W\Sigma_{LHS})=0$.

\begin{example}\label{ex}
Consider the following rank two state $\rho=1/2|\phi_1\rangle\langle \phi_1| + 1/2 |\phi_2\rangle \langle \phi_2|$ where $|\phi_1 \rangle =|0\rangle \otimes \frac{1}{\sqrt{2}}\left (|00\rangle + |11\rangle \right)$, $|\phi_2\rangle =|1\rangle \otimes \frac{1}{\sqrt{2}}\left (|00\rangle - |11\rangle \right) $ and assemblage $\Sigma=\left\{\sigma_{ab|xy}\right\}$ given by $\sigma_{ab|xy}=\mathrm{Tr}_{AB}(P_{a|x}\otimes Q_{b|y}\otimes \mathds{1}\rho)$ with projective measurements $P_{0|0}=Q_{0|0}=|0\rangle \langle 0|$ and  $P_{0|1}=Q_{0|1}=|+\rangle \langle +|$
\begin{equation}
\Sigma=\frac{1}{8}\begin{pmatrix}
\begin{array}{cc|cc}
   2|0\rangle \langle 0| &  2|1\rangle \langle 1| &  2|+\rangle \langle +| &  2|-\rangle \langle -|\\
    2|0\rangle \langle 0| & 2|1\rangle \langle 1|  &  2|-\rangle \langle -| &  2|+\rangle \langle +| \\\hline
    2|0\rangle \langle 0|  &2|1\rangle \langle 1| &\mathds{1} & \mathds{1} \\
		2|0 \rangle \langle 0|  &2|1\rangle \langle 1| &\mathds{1}  & \mathds{1} 
\end{array}
\end{pmatrix}.
\end{equation}It can be easily seen that this assemblage is on the edge (compare with Remark \ref{rem}). Moreover, as it has two rows which are equal, $\Sigma=\frac{1}{2}(\Sigma_1+\Sigma_2)$ where $\Sigma_i$ for $i=1$ ($i=2$) is constructed out of $\Sigma$ by exchanging fourth row (third row) with zeros and multiplying third row (fourth row) by a factor of two. As $\Sigma_1\neq\Sigma_2$ given assemblage is not extremal. Note that for $0\leq p \leq 1$ each assemblage $\Sigma_p=p\Sigma_1+(1-p)\Sigma_2$ is still on the edge - this provides an example of a flat region inside the edge.
Define $Z$ as
\begin{equation}
Z =\begin{pmatrix}
\begin{array}{cc|cc}
    |1\rangle \langle 1| &   |0\rangle \langle 0|  &  |-\rangle \langle -| &  |+\rangle \langle +|\\
  |1\rangle \langle 1| &   |0\rangle \langle 0|  &  |+\rangle \langle +| &   |-\rangle \langle -|\\\hline
    |1\rangle \langle 1| &   |0\rangle \langle 0|  &  0 &  0 \\
		 |1\rangle \langle 1| &  |0\rangle \langle 0|  &  0 &  0
\end{array}
\end{pmatrix}.
\end{equation}Looking at symmetries of $Z$ we can see that computation of $\epsilon$ is equivalent to the following minimization problem
$\min_{\rho} \mathrm{Tr}((2|0\rangle \langle 0|+|+\rangle \langle +|) \rho)= \frac{3-\sqrt{5}}{2}$. By fixing this $\epsilon$ we define a witness $W= Z - \frac{3-\sqrt{5}}{8}\mathds{1}_{\oplus}$ detecting $\Sigma_p$.
\end{example}

\section{Edge assemblages with quantum realization}\label{sV}

In this section we will restrict our attention to the case of assemblages $\Sigma=\left\{\sigma_{ab|xy}\right\}$ with $a,b,x,y \in \left\{0,1\right\}$ and $\sigma_{ab|xy}$ acting on $d$ dimensional space. In what follows we will use $A,B,C$ for description of subsystems (where subsystem C will be characterized).

As not all no-signaling assembles in this setting admits quantum realization it is natural to ask whether there exist edge assemblages that are not postquantum - in this section we provide a positive answer to this question. Note the difference with the related case of the edge of the polytope of bipartite no-signaling boxes with binary outcomes/settings (where $S_1$ is given by boxes of local type, and $S_2$ stands for all boxes). In that case, it can be easily shown that any box on the edge must be non-local and extremal in the polytope of all no-signaling boxes (i.e. it must be a PR-box up to relabeling) - any such extreme point does not admit quantum realization \cite{RTHHPRL}. From this, any assemblage obtained by measurements on tripartite state separable in $AB|C$ cut is never on the edge.

\begin{theorem}\label{thm_old}
For any pure tripartite state $|\psi_{ABC}\rangle\in \mathbb{C}^{2}\otimes \mathbb{C}^{2}\otimes \mathbb{C}^{d}$ entangled in a cut $AB|C$, there exists a pair of projective measurements $\left\{P_{a|0}\right\}_{a=0}^1,\left\{P_{a|1}\right\}_{a=0}^1$ on the subsystem A and  $\left\{Q_{a|0}\right\}_{a=0}^1,\left\{Q_{a|1}\right\}_{a=0}^1$ on the subsystem B respectively, such that a no-signaling assemblage $\Sigma=\left\{\sigma_{ab|xy}\right\}$ obtained by $\sigma_{ab|xy}=\mathrm{Tr}_{AB}(P_{a|x}\otimes Q_{b|y}\otimes \mathds{1}|\psi_{ABC}\rangle \langle \psi_{ABC}|)$ is on the edge of the set of no-signaling assemblages. 
\end{theorem}
\begin{proof} Assume that $|\psi_{ABC}\rangle$ is genuinely entangled (i.e. entangled in any bipartite cut).  By result given in \cite{RBRH}, one can adjust projective measurements in such a way that resulting assemblage $\Sigma=\left\{\sigma_{ab|xy}\right\}$ is both extremal and not LHS, thus it is on the edge of no-signaling. To conclude the proof we may without loss of generality assume that $|\psi_{ABC}\rangle=|\varphi_A\rangle|\varphi_{BC}\rangle$ with $|\varphi_{BC}\rangle\in \mathbb{C}^{2}\otimes \mathbb{C}^{d}$ being entangled. Put $P_{a|0}=P_{a|1}$ with $P_{0|0}=|\varphi_A\rangle \langle \varphi_A|$ and define measurements on subsystem B $\left\{Q_{b|0}\right\}_{b=0}^1,\left\{Q_{b|1}\right\}_{b=0}^1$  in such a way that any two operators of the form $\mathrm{Tr}_{B}(Q_{b|y}\otimes \mathds{1}|\varphi_{BC}\rangle \langle \varphi_{BC}|)$ are not proportional (this is always possible - see for example \cite{RBRH}). It is then obvious that one cannot subtract an LHS part from this assemblage (there is no appropriate rectangle), hence $\Sigma$ is on the edge.
\end{proof}

In the case of quantum no-signaling assemblages arises a natural question concerning the maximal rank of initial state $\rho$ from which, by appropriate choice of local measurements, one can obtain assemblage which is on the edge - we address this issue in a three-qubit setting.

\color{black}
\begin{theorem}\label{thm_main}
Let $\rho_{ABC}\in M_2(\mathbb{C})\otimes M_2(\mathbb{C})\otimes M_2(\mathbb{C})$ be a tripartite state. Consider two pairs of POVMs $\left\{M_{a|0}\right\}_{a=0}^1,\left\{M_{a|1}\right\}_{a=0}^1$ on the subsystem A and two pairs of POVMs $\left\{N_{b|0}\right\}_{b=0}^1,\left\{N_{b|1}\right\}_{b=0}^1$ on the subsystem B. Define a no-signaling assemblage via $\sigma_{ab|xy}=\mathrm{Tr}_{AB}(M_{a|x}\otimes N_{b|y}\otimes \mathds{1}\rho_{ABC})$. If $\mathrm{rank}(\rho_{ABC})\geq 3$ then $\Sigma=\left\{\sigma_{ab|xy}\right\}$ is not on the edge of the set of no-signaling assemblages. 
\end{theorem}

\begin{sproof}
We may restrict to the case $\mathrm{rank}(\rho_{ABC})=3$ (see Appedix). Assume that all POVMs are given by PVMs. 
Note that certain structures (defined by certain positions occupied by operators with assumed ranks) of assemblage may be excluded as $\mathrm{rank}(\rho_{ABC})=3$ provides constraints on ranks for operators obtained after measurements on one of (or both) subsystems. Careful analysis of all remaining and relevant structures shows (see Appedix) that in each case one can either subtract an LHS part or considered structure is impossible when $\mathrm{rank}(\rho_{ABC})=3$. Finally, one can see (Appendix) that any assemblage $\Sigma$ with POVMs realization can be written as a nontrivial convex combination $\Sigma=p\Sigma_1+(1-p)\Sigma_2$ where $\Sigma_1$ admits a quantum realization with PVMs.
\end{sproof}

Theorem \ref{thm_main} may be compared with a case of the edge of quantum states where the set of all three-qubit states plays a role of $S_2$ and a subset of all fully separable three-qubit states in natural way stands for $S_1$. Note that there is no product vector $|\phi_{A}\rangle|\phi_{B}\rangle|\phi_{C}\rangle$ in the image of a state $\tau_{ABC}=\frac{1}{4}\left(\mathds{1}-\Pi\right)$ where $\Pi$ stands for projection on space spanned by four vectors forming an unextendible product basis (UPB) \cite{UPB03}. Therefore, $\tau_{ABC}$ is on the edge but on the other hand $\mathrm{rank}(\tau_{ABC})=4$.

Note that due to Example \ref{ex}, rank three is the smallest rank for which theorem like the one above may be formulated, since there is a choice of states with rank two from which one can obtain assemblages on the edge. The following theorem provides sufficient conditions for such situations.

\begin{theorem}\label{thm_main2}
Let $\rho_{ABC}\in M_2(\mathbb{C})\otimes M_2(\mathbb{C})\otimes M_d(\mathbb{C})$ with $\mathrm{rank}(\rho_{ABC})=2$. Assume that there exists a POVM $\left\{M_{a|x}\right\}_{a=0}^1$ on the subsystem A such that both $\rho_{a|x}=\mathrm{Tr}_{A}(M_{a|x}\otimes \mathds{1}\otimes \mathds{1}\rho_{ABC})$ are rank one and entangled. Then one can choose additional projective measurements such that $\Sigma$ given by $\sigma_{ab|xy}=\mathrm{Tr}_{AB}(M_{a|x}\otimes Q_{b|y}\otimes \mathds{1}\rho_{ABC})$ is on the edge. 
\end{theorem}

\begin{proof}Since $\rho_{a|0}$ are pure and entangled for $a=0,1$, one can define a pair of nontrivial projective measurements $\left\{Q_{b|0}\right\}_{b=0}^1,\left\{Q_{b|1}\right\}_{b=0}^1$ on subsystem B in such a way, that for a given $a$ any two operators of the form $\sigma_{ab|0y}=\mathrm{Tr}_{B}(Q_{b|y}\otimes \mathds{1}\rho_{a|0})$ are not proportional (this is always possible - see for example \cite{RBRH}) - this concludes the proof.
\end{proof}

\section{Discussion}\label{sVI}

Motivated by the relation between sets of separable states and all states we have introduced the notion of the edge of no-signaling assemblages constructed concerning the subset of assemblages that admit an LHS model. We have provided a full characterization of edge assemblages and we have stated some necessary criteria for being on the edge. Moreover, we have discussed the notion of witnesses for no-signaling assemblages beyond LHS realization and we have related this discussion to the previously evoked concept of the edge. In the case of tripartite states (with two qubits), we have characterized pure states steerable to edge assemblages and we have provided sufficient conditions for this when states have rank $2$. Finally, we have proved a no-go type result for three-qubit states with a rank greater than or equal to $3$.

While presented results better describe the set of no-signaling assemblages, many interesting questions are remaining. In particular, it would be natural to ask for a full characterization of witnesses, at least in a case of a specific setting. Besides, further generalizations of Theorem \ref{thm_main} may be performed by taking arbitrary dimension $d$ or changing the number of uncharacterized parties or possible settings and outcomes.

The concept of the edge of the no-signaling assemblages represents - in analogy to the standard entanglement picture - the strongest form of correlations within the set of all no-signaling assemblages since edge assemblages do not admit any admixture of local hidden states. We hope that because of the latter unique property they will be an important resource for the typical security tasks in semi-device independent quantum information. The explicit analysis presented here should also provide tools for the construction of practical witnesses of the edge assemblages in the future.

\begin{acknowledgments}
\textit{Acknowledgments}- M.B., R.R.R., and P.H. acknowledge support by the Foundation for Polish Science (IRAP project, ICTQT, contract no. 2018/MAB/5, co-financed by EU within Smart Growth Operational Programme). M.B. is grateful to Marcin Marciniak for the discussion.
\end{acknowledgments}

\appendix
\section{Proof of Theorem \ref{thm_main}}

To adjust convection, consider two pairs of POVMs with binary outputs $\left\{M_{a|0}\right\}_{a=0}^1,\left\{M_{a|1}\right\}_{a=0}^1$. We say that they are the same up to relabeling if $M_{0|0}=M_{0|1}$ or $M_{0|0}=M_{1|1}$. In other case we say that POVMs are different up to relabeling. In particular the same convention will be used for pairs of projective measurements (PVMs).

In order to prove the no-go results of Theorem \ref{thm_main}, we will need to introduce a series of simple lemmas. In what follows we will use observations stated in discussion after Corollary \ref{corollary} (in the main text) - possibility of subtraction of an LHS part will be stated in terms of the existence of certain rectangles within considered assemblages.

\begin{lem}\label{lem1rank3}Let $\rho_{AB}\in M_2(\mathbb{C})\otimes M_d(\mathbb{C})$ be a bipartite state and $\rho_a=\mathrm{Tr}_{A}(P_a\otimes \mathds{1}\rho_{AB})$ for some nontrivial projective measurement $\left\{P_a\right\}_{a=0}^1$ on the subsystem A. If $\mathrm{rank}(\rho_1)=\mathrm{rank}(\rho_{2})=1$, then $\mathrm{rank}(\rho_{AB})\neq 3$. Moreover, if $d\geq3$, $\mathrm{rank}(\rho_{AB})=3$ and $\mathrm{rank}(\rho_1)=0$ imply $\mathrm{rank}(\rho_{2})=3$ and if $d=2$ none of $\rho_{a}$ have rank zero. 
\end{lem}

\begin{proof} Without loss of generality we may put $P_0=|0\rangle \langle 0|$ and $P_1=|1\rangle \langle 1|$. Assume that on the contrary $\mathrm{rank}(\rho_{AB})=3$. Observe that $\rho_{AB}$ can be expressed by spectral decomposition
\begin{equation}\label{eq_1}
\rho_{AB}=\sum_{i=1}^3\lambda_i|\psi_i\rangle \langle \psi_i|
\end{equation}with 
\begin{equation}\label{eq_2}
|\psi_i\rangle=\alpha_i|0\rangle|h_i\rangle+\beta_i|1\rangle|g_i\rangle
\end{equation}where $\lambda_i>0$ and $|\psi_i\rangle, |h_i\rangle, |g_i\rangle$ are normalized for any $i$. If $\mathrm{rank}(\rho_1)=\mathrm{rank}(\rho_{2})=1$, then for any $i=1,2,3$ we have $|h_i\rangle=|h\rangle$ and $|g_i\rangle=|g\rangle$ (since possible phase factors may be always incorporated in $\alpha_i$ or $\beta_i$ respectively and if $\alpha_i$ or $\beta_i$ is equal to zero form of $|h_i\rangle$ or $|g_i\rangle$ may be arbitrary). Therefore, all pure states (\ref{eq_2}) belong to the same two-dimensional subspace. Since $\left\{|\psi_i\rangle\right\}_i$ is an orthonormal system we obtain a contradiction - this concludes the first part of the proof. 

Now assume once more that $\mathrm{rank}(\rho_{AB})=3$ and $d\geq 3$. If $\mathrm{rank}(\rho_1)=0$ then by formula (\ref{eq_1}) and (\ref{eq_2}) we get that each $|\psi_i\rangle$ is proportional to $|1\rangle|g_i\rangle$. Orthogonality of $\left\{|\psi_i\rangle\right\}_i$ imply then orthogonality of $\left\{|g_i\rangle\right\}_i$ (it is possible if and only if $d\geq 3$). Since 
\begin{equation}
\rho_{2}=\sum_{i=1}^3\lambda_i|g_i\rangle \langle g_i|
\end{equation}we have $\mathrm{rank}(\rho_{2})=3$.
\end{proof}

\begin{lem}\label{lem2rank2} Let $\rho_{AB}\in M_2(\mathbb{C})\otimes M_d(\mathbb{C})$ be a bipartite state such that $\mathrm{rank}(\rho_{AB})=2$. Consider two pairs of nontrivial projective measurements $\left\{P_{a|0}\right\}_{a=0}^1,\left\{P_{a|1}\right\}_{a=0}^1$ on the subsystem A and define $\rho_{a|x}=\mathrm{Tr}_{A}(P_{a|x}\otimes\mathds{1}\rho_{AB})$. Assume that $\left\{P_{a|0}\right\}_{a=0}^1$ and $\left\{P_{a|1}\right\}_{a=0}^1$ are different (up to relabeling). If one of $\rho_{a|x}$ has rank zero then the others have rank two and they are proportional. If $x_1\neq x_2$ and $\rho_{0|x_1},\rho_{1|x_1}$ have rank one then either $\rho_{0|x_2}$ and $\rho_{1|x_2}$ have rank two or all $\rho_{a|x}$ have rank one and they are proportional.
\end{lem}
\begin{proof}
Without loss of generality we may put $P_{0|0}=|0\rangle \langle 0|$ and $P_{1|0}=|1\rangle \langle 1|$. Assume that $\mathrm{rank}(\rho_{AB})=2$. Observe that $\rho_{AB}$ can be expressed by spectral decomposition
\begin{equation}
\rho_{AB}=\sum_{i=1}^2\lambda_i|\psi_i\rangle \langle \psi_i|
\end{equation}with 
\begin{equation}
|\psi_i\rangle=\alpha_i|0\rangle|h_i\rangle+\beta_i|1\rangle|g_i\rangle
\end{equation}where $\lambda_i>0$ and $|\psi_i\rangle, |h_i\rangle, |g_i\rangle$ are normalized for any $i$.

Firstly, without loss of generality assume that $\rho_{0|0}=0$, then $\alpha_1=\alpha_2=0$ and $|g_1\rangle \perp |g_2\rangle$. If so then the other $\rho_{a|x}$ are proportional to $\sum_{i=1}^2\lambda_i|g_i\rangle \langle g_i|$ which has rank two.

 If without loss of generality $\rho_{0|0}, \rho_{1|0}$ have rank one, then we can put $|h_i\rangle =|h\rangle$ and $|g_i\rangle =|g\rangle$. If $|h\rangle$ and $|g\rangle$ are equal (up to irrelevant phase), then $\rho_{0|1}, \rho_{1|1}$ are proportional to $|g\rangle \langle g|$. If this is not the case then $\rho_{0|1}, \rho_{1|1}$ have rank two.
\end{proof}

\begin{lem}\label{lem45}

Consider a collection of orthonormal vectors $|c_i\rangle \in \mathbb{C}^d$ and a two-dimensional subspace $H\subset \mathbb{C}^d$. Projection of each $|c_i\rangle$ on $H$ is proportional to the same vector if and only if there exists a vector $|c\rangle \in H$ othonormal to each $|c_i\rangle$. 
\end{lem}

\begin{proof} If $|c\rangle \in H$, then projection of each $|c_i\rangle$ is proportional to normalized $|c_{\perp}\rangle \in H$, where $|c_{\perp}\rangle\perp |c\rangle$.

On the other and, let $Q$ be a projection on $H$ and let $Q|c_i\rangle=\alpha_i|c_{\perp}\rangle$ with $\alpha_i\geq 0$ for some normalized vector $|c_{\perp}\rangle \in H$ and all $i$. Take normalized vector $|c\rangle \in H$ such that $|c_{\perp}\rangle\perp |c\rangle$. Then $Q=|c\rangle\langle c| + |c_{\perp}\rangle\langle c_{\perp}|$ and $\langle c|c_i\rangle=0$ for all $i$.
\end{proof}

\begin{lem}\label{lem5rank3} Let $\rho_{ABC}\in M_2(\mathbb{C})\otimes M_2(\mathbb{C})\otimes M_2(\mathbb{C})$ be a tripartite state. Consider two pairs of nontrivial projective measurements $\left\{P_{a|0}\right\}_{a=0}^1,\left\{P_{a|1}\right\}_{a=0}^1$ on the subsystem A and two pairs (different up to relabeling) of nontrivial projective measurements $\left\{Q_{b|0}\right\}_{b=0}^1,\left\{Q_{b|1}\right\}_{b=0}^1$ on the subsystem B. Define a no-signaling assemblage via  $\sigma_{ab|xy}=\mathrm{Tr}_{AB}(P_{a|x}\otimes Q_{b|y}\otimes \mathds{1}\rho_{ABC})$. If $\mathrm{rank}(\rho_{ABC})= 3$ then it is impossible that $\Sigma=\left\{\sigma_{ab|xy}\right\}$ takes one of the following forms
\begin{equation}\label{form}
\begin{pmatrix}
\begin{array}{cc|cc}
\phi_{00}&\phi_{01}&   X&   X\\
\phi_{10}& \phi_{11}&X&X \\ \hline
X& X&\varphi_1& \varphi_2\\
X & X&\varphi_3&   \varphi_4
\end{array}
\end{pmatrix}
\end{equation},
\begin{equation}\label{form2}
\begin{pmatrix}
\begin{array}{cc|cc}
\phi_{00}&X&   \varphi_1&   X\\
\phi_{10}& X&\varphi_3&X \\ \hline
X& \phi_{01}& X & \varphi_2\\
X & \phi_{11}&X &   \varphi_4
\end{array}
\end{pmatrix},
\end{equation}where $X$ stands for any rank two operator and $\phi_i$ stands for rank one operator proportional to $|\phi_i\rangle \langle \phi_i|$.
\end{lem}

\begin{proof}Without loss of generality we may put $P_{0|0}=Q_{0|0}=|0\rangle \langle 0|$, $P_{1|0}=Q_{1|0}=|1\rangle \langle 1|$ and $P_{0|1}=|0'\rangle \langle 0'|$, $P_{1|1}=|1'\rangle \langle 1'|$, $Q_{0|1}=|0''\rangle \langle 0''|$, $Q_{1|1}=|1''\rangle \langle 1''|$ where 
\begin{equation}\nonumber
|0'\rangle=\overline{\alpha} |0\rangle +\overline{\beta} |1\rangle,\ |1'\rangle=\beta |0\rangle -\alpha |1\rangle,
\end{equation}
\begin{equation}\nonumber
|0''\rangle=\overline{\gamma} |0\rangle +\overline{\delta} |1\rangle,\ |1''\rangle=\delta |0\rangle -\gamma |1\rangle.
\end{equation}Note that $\gamma,\delta\neq 0$ due to general assumption (different pair of PVMs on the subsystem B). Observe that since $\mathrm{rank}(\rho_{ABC})=3$ we have a spectral decomposition
\begin{equation}\label{spectral}
\rho_{ABC}=\sum_{i=1}^3\lambda_i|\psi_i\rangle \langle \psi_i|.
\end{equation}Finally, define vectors $|x\rangle, |y\rangle\in \mathbb{C}^4$ via equations
\begin{equation}
|\phi_{kl}\rangle=x_{kl}|0\rangle+y_{kl}|1\rangle.
\end{equation}Observe that for any pair $k,l$ coefficients $x_{kl},y_{kl}$ cannot be both equal to zero (due to normalization).

In the first part of the proof, assume on the contrary that $\Sigma$ is of the form (\ref{form}). If so, then $\alpha,\beta,\gamma,\delta\neq 0$ and $|\psi_i\rangle$ from (\ref{spectral}) have the following form
\begin{equation}\nonumber
|\psi_i\rangle=c^i_{00}|00\rangle|\phi_{00}\rangle+c^i_{01}|01\rangle|\phi_{01}\rangle+c^i_{10}|10\rangle|\phi_{10}\rangle+c^i_{11}|11\rangle|\phi_{11}\rangle,
\end{equation}
where $|c_i\rangle=(c^i_{00}, c^i_{01}, c^i_{10}, c^i_{11})\in \mathbb{C}^4$ are normalized and $|c_i\rangle\perp |c_j\rangle$ for $i\neq j$.  Obviously $\mathrm{dim}\left(\mathrm{span}\left\{|x\rangle, |y\rangle\right\}\right)=2$ (the last statement is true because by the no-signaling condition for the first row of (\ref{form}), $|\phi_{00}\rangle$ and $|\phi_{01}\rangle$ are not proportional, so $(x_{00},x_{01}), (y_{00},y_{01})$ are linearly independent). Finally, let us introduce operators $L_j$ for $j=1,2,3,4$
\begin{equation}\nonumber
L_1=\mathrm{diag}(\alpha\gamma, \alpha\delta, \beta\gamma,\beta\delta),\ L_2=\mathrm{diag}(\alpha\overline{\delta}, -\alpha\overline{\gamma},\beta\overline{\delta}, -\beta\overline{\gamma}),
\end{equation}
\begin{equation}\nonumber
L_3=\mathrm{diag}(\overline{\beta}\gamma, \overline{\beta}\delta, -\overline{\alpha}\gamma,-\overline{\alpha}\delta),\ L_4=\mathrm{diag}(\overline{\beta}\overline{\delta}, -\overline{\beta}\overline{\gamma}, -\overline{\alpha}\overline{\delta},\overline{\alpha}\overline{\gamma}).
\end{equation}Note that operators $L_j$ represents action of measurements $P_{a|1}, Q_{b|1}$ (related to positions in (\ref{form}) occupied by $\varphi_j$) on vectors $|\psi_i\rangle$. Indeed, for example we have
\begin{equation}
\mathrm{Tr}_{AB}(P_{0|1}\otimes Q_{0|1}\otimes \mathds{1}|\psi_i\rangle \langle \psi_i|)=|\tilde{\varphi}_{1,i}\rangle \langle \tilde{\varphi}_{1,i}|
\end{equation}where
\begin{widetext}
\begin{equation}\label{obs1}
|\tilde{\varphi}_{1,i}\rangle=c^i_{00}\alpha\gamma|\phi_{00}\rangle+c^i_{01}\alpha\delta|\phi_{01}\rangle+c^i_{10}\beta\gamma|\phi_{10}\rangle+c^i_{11}\beta\delta|\phi_{11}\rangle=\langle \overline{c}_i|L_1|x\rangle |0\rangle+\langle \overline{c}_i|L_1|y\rangle |1\rangle
\end{equation}and similarly
\begin{equation}\label{obs2}
|\tilde{\varphi}_{2,i}\rangle=c^i_{00}\alpha\overline{\delta}|\phi_{00}\rangle-c^i_{01}\alpha\overline{\gamma}|\phi_{01}\rangle+c^i_{10}\beta\overline{\delta}|\phi_{10}\rangle-c^i_{11}\beta\overline{\gamma}|\phi_{11}\rangle=\langle \overline{c}_i|L_2|x\rangle |0\rangle+\langle \overline{c}_i|L_2|y\rangle |1\rangle,
\end{equation}
\begin{equation}\label{obs3}
|\tilde{\varphi}_{3,i}\rangle=c^i_{00}\overline{\beta}\gamma|\phi_{00}\rangle+c^i_{01}\overline{\beta}\delta|\phi_{01}\rangle-c^i_{10}\overline{\alpha}\gamma|\phi_{10}\rangle-c^i_{11}\overline{\alpha}\delta|\phi_{11}\rangle=\langle \overline{c}_i|L_3|x\rangle |0\rangle+\langle \overline{c}_i|L_3|y\rangle |1\rangle,
\end{equation}
\begin{equation}\label{obs4}
|\tilde{\varphi}_{4,i}\rangle=c^i_{00}\overline{\beta}\overline{\delta}|\phi_{00}\rangle-c^i_{01}\overline{\beta}\overline{\gamma}|\phi_{01}\rangle-c^i_{10}\overline{\alpha}\overline{\delta}|\phi_{10}\rangle+c^i_{11}\overline{\alpha}\overline{\gamma}|\phi_{11}\rangle=\langle \overline{c}_i|L_4|x\rangle |0\rangle+\langle \overline{c}_i|L_4|y\rangle |1\rangle,
\end{equation}
\end{widetext}with $|\overline{c}_i\rangle=(\overline{c}^i_{00}, \overline{c}^i_{01}, \overline{c}^i_{10}, \overline{c}^i_{11})$. By assumption we see that operators in the intersection of the third and the fourth rows with the third and the fourth columns in (\ref{form}) are rank one, so for any fixed $j=1,2,3,4$, all $|\tilde{\varphi}_{j,i}\rangle$ are proportional to some $|\varphi_{j}\rangle$ (some but not all of them may be equal to zero). Observe that by (\ref{obs1}-\ref{obs4}) this is true if for any fixed $j=1,2,3,4$, projection of each $|\overline{c}_i\rangle$ on $\mathrm{span}\left\{L_j|x\rangle,L_j|y\rangle\right\}$ is the same up to normalization or equal to zero. From Lemma \ref{lem45} we see that this is true only if there exists a common vector $|c_4 \rangle\in \mathbb{C}^4$ orthonormal to others $|\overline{c}_i\rangle$ and such that $|c_4 \rangle\in \mathrm{span}\left\{L_j|x\rangle, L_j|y\rangle\right\}$ for all $j$ (note that this is true since $|c_4 \rangle$ is unique up to phase due to orthogonality condition in $\mathbb{C}^4$).

To conclude the proof it is enough to show that it is impossible to find a vector $|c_4 \rangle$ with such property that $L_j^{-1}|c_4 \rangle\in \mathrm{span}\left\{|x\rangle, |y\rangle\right\}$ for all $j$. Put $|c_4 \rangle=(c_1,c_2,c_3,c_4)$. Observe that if $|c_4 \rangle$ has three coefficients equal to zero, then all $|c_i \rangle$ for $i=1,2,3$ have a common zero coefficient, i.e. there is a pair $k,l$ such that $c_{kl}^i=0$ for all $i=1,2,3$. Therefore, $|c_4 \rangle$ has at most two coefficients equal to zero. 

Assume first that $|c_4 \rangle$ has no zero coefficient. Then $L_j^{-1}|c_4 \rangle$ for $j=1,2,3$ form a linearly independent set. Indeed, we have 
\begin{widetext}
\begin{equation}\nonumber
\mathrm{det}\begin{pmatrix}
\begin{array}{ccc}
\alpha^{-1}\gamma^{-1}c_1 & \alpha^{-1}\overline{\delta}^{-1}c_1  &  \overline{\beta}^{-1}\gamma^{-1}c_1\\
\alpha^{-1}\delta^{-1}c_2 & -\alpha^{-1}\overline{\gamma}^{-1}c_2  & \overline{\beta}^{-1}\delta^{-1} c_2\\
\beta^{-1}\gamma^{-1}c_3 &  \beta^{-1}\overline{\delta}^{-1}c_3  & -\overline{\alpha}^{-1}\gamma^{-1}c_3
\end{array}
\end{pmatrix}=c_1c_2c_3\alpha^{-1}\gamma^{-1}\left(|\alpha|^{-2}+|\beta|^{-2}\right)\left(|\gamma|^{-2}+|\delta|^{-2}\right)\neq 0.
\end{equation}
\end{widetext}
Therefore, $L_j^{-1}|c_4 \rangle\in \mathrm{span}\left\{|x\rangle, |y\rangle\right\}$ cannot be true for all $j=1,2,3$, as $\mathrm{dim}\left(\mathrm{span}\left\{|x\rangle, |y\rangle\right\}\right)=2$. This leads to a contradiction. 

Assume then that $|c_4 \rangle$ has one or two coefficients equal to zero - for example let $c_1,c_2\neq 0$ and $c_3=0$. then we have
\begin{widetext}
\begin{equation}\nonumber
\mathrm{det}\begin{pmatrix}
\begin{array}{ccc}
\alpha^{-1}\gamma^{-1}c_1 & \alpha^{-1}\overline{\delta}^{-1}c_1 \\
\alpha^{-1}\delta^{-1}c_2 & -\alpha^{-1}\overline{\gamma}^{-1}c_2
\end{array}
\end{pmatrix}=-c_1c_2\alpha^{-2}\left(|\gamma|^{-2}+|\delta|^{-2}\right)\neq 0,
\end{equation}
\end{widetext}so in particular $L_1^{-1}|c_4 \rangle$ and $L_2^{-1}|c_4 \rangle$ are linearly independent. Note that if $L_j^{-1}|c_4 \rangle \in \mathrm{span}\left\{|x\rangle, |y\rangle\right\}$ for all $j=1,2$, then $\mathrm{span}\left\{L_1^{-1}|c_4 \rangle, L_2^{-1}|c_4 \rangle\right\}=\mathrm{span}\left\{|x\rangle, |y\rangle\right\}$ and by $c_3=0$ we have $x_{10}=y_{10}=0$ which is a contradiction with normalization of $|\phi_{10}\rangle$. By similar calculations based on properties of $L_j^{-1}$, the same conclusion can be shown for any $|c_4 \rangle$ with one or two coefficients equal to zero. In the end, by negation of condition from lemma \ref{lem45}, it is impossible to obtain assemblage of the form (\ref{form}) starting from nontrivial projective measurement on rank three state $\rho_{ABC}$.

To conclude the proof, consider a second possibility, namely assume that $\Sigma$ is of the form (\ref{form2}). If so, then $\alpha,\beta,\gamma,\delta\neq 0$. Define a unitary $U_{AB}$ acting on subsystem AB 
\begin{equation}
U_{AB}=\mathds{1}_A\otimes|0\rangle \langle 0|+ U_A\otimes|1\rangle \langle 1|
\end{equation}with $U_A$ being a unitary such that $U^\dagger_A|0\rangle=|0'\rangle, U^\dagger_A|1\rangle=|1'\rangle$. Obviously, $\tilde{\rho}_{ABC}=U_{AB}\rho_{ABC}U^\dagger_{AB}$ is still a state with rank equal to three. By the properties of $U_A$ and assumed form of $\Sigma$, one can see that $\tilde{\Sigma}=\left\{\tilde{\sigma}_{ab|xy}\right\}$ defined by $\tilde{\sigma}_{ab|xy}=\mathrm{Tr}_{AB}(P_{a|x}\otimes Q_{b|y}\otimes \mathds{1}\tilde{\rho}_{ABC})$ have the following form
\begin{equation}
\begin{pmatrix}
\begin{array}{cc|cc}
\phi_{00}&\phi_{01}&   \ldots&    \ldots\\
\phi_{10}& \phi_{11}& \ldots& \ldots \\ \hline
X&\ldots &  \ldots &  \ldots\\
X &\ldots & \ldots &    \ldots
\end{array}
\end{pmatrix}.
\end{equation}
If so, then similarly to the first part of the proof, one can write a spectral decomposition $\tilde{\rho}_{ABC}=\sum_{i=1}^3\lambda_i|\tilde{\psi}_i\rangle \langle \tilde{\psi}_i|$ with $|\tilde{\psi}_i\rangle=U_{AB}|\psi_i\rangle$ and 
\begin{equation}\nonumber
|\tilde{\psi}_i\rangle=c^i_{00}|00\rangle|\phi_{00}\rangle+c^i_{01}|01\rangle|\phi_{01}\rangle+c^i_{10}|10\rangle|\phi_{10}\rangle+c^i_{11}|11\rangle|\phi_{11}\rangle,
\end{equation}where $|c_i\rangle$ for $i=1,2,3$ are like in the first part of the proof. Using once more properties of $U_A$, we see that orthonormal vectors from spectral decomposition of initial state $\rho_{ABC}$ (see (\ref{spectral})) can be given as 
\begin{widetext}
\begin{equation}\nonumber
|\psi_i\rangle=c^i_{00}|00\rangle|\phi_{00}\rangle+c^i_{01}\overline{\alpha}|01\rangle|\phi_{01}\rangle+c^i_{11}\beta|01\rangle|\phi_{11}\rangle +c^i_{10}|10\rangle|\phi_{10}\rangle+c^i_{01}\overline{\beta}|11\rangle|\phi_{01}\rangle-c^i_{11}\alpha|11\rangle|\phi_{11}\rangle  
\end{equation}
\end{widetext}
We can represents action of appropriate measurements $P_{a|x}, Q_{b|y}$ (related to positions in (\ref{form2}) occupied by $\varphi_j$) on vectors $|\psi_i\rangle$ by introduction of vectors $|\tilde{\varphi}_{j,i}\rangle$ given as
\begin{widetext}
\begin{equation}
|\tilde{\varphi}_{1,i}\rangle=c^i_{00}\gamma|\phi_{00}\rangle+c^i_{01}\delta\overline{\alpha}|\phi_{01}\rangle+c^i_{11}\delta\beta|\phi_{11}\rangle=\langle \overline{c}_i|S_1|x\rangle |0\rangle+\langle \overline{c}_i|S_1|y\rangle |1\rangle.
\end{equation}
\begin{equation}
|\tilde{\varphi}_{2,i}\rangle=c^i_{00}\overline{\delta}\alpha|\phi_{00}\rangle-c^i_{01}\overline{\gamma}|\phi_{01}\rangle+c^i_{10}\overline{\delta}\beta|\phi_{10}\rangle=\langle \overline{c}_i|S_2|x\rangle |0\rangle+\langle \overline{c}_i|S_2|y\rangle |1\rangle.
\end{equation}
\begin{equation}
|\tilde{\varphi}_{3,i}\rangle=c^i_{01}\delta\overline{\beta}|\phi_{01}\rangle+c^i_{10}\gamma|\phi_{10}\rangle-c^i_{11}\delta\alpha|\phi_{11}\rangle=\langle \overline{c}_i|S_3|x\rangle |0\rangle+\langle \overline{c}_i|S_3|y\rangle |1\rangle.
\end{equation}
\begin{equation}
|\tilde{\varphi}_{4,i}\rangle=c^i_{00}\overline{\delta}\overline{\beta}|\phi_{00}\rangle-c^i_{10}\overline{\delta}\overline{\alpha}|\phi_{10}\rangle-c^i_{11}\overline{\gamma}|\phi_{11}\rangle=\langle \overline{c}_i|S_4|x\rangle |0\rangle+\langle \overline{c}_i|S_4|y\rangle |1\rangle.
\end{equation}
\end{widetext}where
\begin{equation}\nonumber
S_1=\mathrm{diag}(\gamma, \delta\overline{\alpha},0,\delta\beta),\ S_2=\mathrm{diag}(\overline{\delta}\alpha,-\overline{\gamma},\overline{\delta}\beta,0),
\end{equation}
\begin{equation}\nonumber
S_3=\mathrm{diag}(0,\delta\overline{\beta},\gamma, -\delta\alpha),\ S_4=\mathrm{diag}(\overline{\delta}\overline{\beta},0, -\overline{\delta}\overline{\alpha},-\overline{\gamma}).
\end{equation}

Obviously (\ref{form2}) implies that for any fixed $j=1,2,3,4$, all $|\tilde{\varphi}_{j,i}\rangle$ are proportional to some $|\varphi_{j}\rangle$ (some but not all of them may be equal to zero). Due to no-signaling conditions for the first and the second column in (\ref{form2}), $|\phi_{00}\rangle$ and $|\phi_{10}\rangle$ are linearly independent and the same is true for $|\phi_{01}\rangle$ and $|\phi_{11}\rangle$. If so, then pairs $(x_{00},x_{10}), (y_{00},y_{10})$ and $(x_{01},x_{11}), (y_{01},y_{11})$ are linearly independent as well. Note that from this, for each $j=1,2,3,4$, $\mathrm{dim}(\mathrm{span}\left\{S_j|x\rangle, S_j|y\rangle\right\})=2$. Therefore, we see that the structure of (\ref{form2}) implies (according to Lemma \ref{lem45}) that (unique up to the phase in $\mathbb{C}^4$) vector $|c_4 \rangle$ orthonormal to each $|\overline{c}_i \rangle$ for $i=1,2,3$ must satisfy $|c_4 \rangle\in \mathrm{span}\left\{S_j|x\rangle, S_j|y\rangle\right\}$ for any $j=1,2,3,4$. However, by the form of diagonal matrices $S_j$, $|c_4 \rangle=(0,0,0,0)$ which is a contradiction and the second part of the proof is concluded.
\end{proof} \\

Now we are ready to present a proof of Theorem \ref{thm_main}.

\begin{proof} It is enough to show the statement in the case of $\mathrm{rank}(\rho_{ABC})=3$. Indeed, let $\mathrm{rank}(\rho_{ABC})> 3$, then $\rho_{ABC}=\alpha\rho_{ABC}^{(1)}+\beta\rho_{ABC}^{(2)}$ where $\alpha,\beta>0$ and $\rho_{ABC}^{(1)},\rho_{ABC}^{(2)}$ are states such that $\mathrm{rank}(\rho_{ABC}^{(1)})=3$. Observe that 
$\sigma_{ab|xy}=\alpha\sigma^{(1)}_{ab|xy}+\beta\sigma^{(2)}_{ab|xy}$ with $\sigma^{(i)}_{ab|xy}=\mathrm{Tr}_{AB}(M_{a|x}\otimes N_{b|y}\otimes \mathds{1}\rho_{ABC}^{(i)})$ for $i=1,2$. Therefore, if any assemblage obtained from rank three state in not on the edge, the same must be true for assemblage obtained from state with higher rank. Therefore we set $\mathrm{rank}(\rho_{ABC})=3$.

From now on unless stated otherwise, let us, instead of general POVMs, consider only nontrivial projective measurements with elements $P_{a|x},Q_{b|y}$ (different from $0$ or $\mathds{1}$) acting on subsystem A and B respectively. Therefore we are interested in structural analysis of no-signaling assemblage $\Sigma=\left\{\sigma_{ab|xy}\right\}$  given by $\sigma_{ab|xy}=\mathrm{Tr}_{AB}(P_{a|x}\otimes Q_{b|y}\otimes \mathds{1}\rho_{ABC})$. For two given pairs of nontrivial projective measurement $\left\{P_{a|0}\right\}_{a=0}^1,\left\{P_{a|1}\right\}_{a=0}^1$ on the subsystem $A$ we define $\rho_{a|x}=\mathrm{Tr}_{A}(P_{a|x}\otimes\mathds{1} \otimes \mathds{1} \rho_{ABC})$ and we put $r_1=\mathrm{rank}(\rho_{0|0}), r_2=\mathrm{rank}(\rho_{1|0}), r_3=\mathrm{rank}(\rho_{0|1}), r_4=\mathrm{rank}(\rho_{1|1})$. In similar way we introduce $\pi_{b|y}=\mathrm{Tr}_{B}(\mathds{1}\otimes Q_{b|y}\otimes\mathds{1}  \rho_{ABC})$ for measurements on the subsystem B and $q_i, i=1,2,3,4$ for related ranks of this operators. In principal $r_1,r_2,r_3,r_4\in \left\{0,1,2,3\right\}$ and the same is true for $q_i$. Without loss of generality we assume that $r_1\geq r_2$ and $r_3\geq r_4$ (it is possible due to relabeling). According to Lemma \ref{lem1rank3} we may restrict our attention to $(r_1,r_2), (r_3,r_4)$ of the form $(3,0), (3,1),(3,2),(3,3), (2,1),(2,2)$. 

In the remaining part of the proof, we will show that all forms of $\Sigma$ compatible with all possible values of $r_i,q_j$ (i.e. not excluded by $\mathrm{rank}(\rho_{ABC})=3$) are either not on the edge or are simply impossible to obtain (note that some of the considered structures, from which one can subtract an LHS part, could still be impossible to obtain due to some reasons which are not discussed in this proof - this is, however, irrelevant for the considered question, if a given structure is not an obstruction for the subtraction of an LHS part, then its existence is not important).

Observe that if one can show that a given assemblage, defined with measurements $\left\{Q_{b|0}\right\}_{b=0}^1,\left\{Q_{b|1}\right\}_{b=0}^1$  is not on the edge, then assemblage, defined with the same $\left\{Q_{b|0}\right\}_{b=0}^1$ and  $\left\{Q_{b|1}\right\}_{b=0}^1$ equal to $\left\{Q_{b|0}\right\}_{b=0}^1$ (up to relabeling), must be outside the edge as well (the same observation is true for measurements on the subsystem A), since the initial rectangle from which one can subtract an LHS part will be changed into rectangle with positions occupied by two of operators presented in the previous rectangle - as it can be seen in the following example
\begin{equation}\nonumber
\begin{pmatrix}
\begin{array}{cc|cc}
 \color{red}\sigma_{00|00} \color{black} &  \ldots &  \color{red}\sigma_{00|01}\color{black} & \ldots  \\  
\ldots & \ldots& \ldots  & \ldots \\ \hline
 \ldots & \ldots & \ldots&  \ldots \\
    \color{red}\sigma_{10|10}\color{black} & \ldots &  \color{red}\sigma_{10|11}\color{black} & \ldots
\end{array}
\end{pmatrix} \rightarrow 
\begin{pmatrix}
\begin{array}{cc|cc}
 \color{red}\sigma_{00|00} \color{black} &  \ldots &  \color{red}\sigma_{00|00}\color{black} & \ldots  \\  
\ldots & \ldots& \ldots  & \ldots \\ \hline
 \ldots & \ldots & \ldots&  \ldots \\
    \color{red}\sigma_{10|10}\color{black} & \ldots &  \color{red}\sigma_{10|10}\color{black} & \ldots
\end{array}
\end{pmatrix}.
\end{equation}Therefore, if we show that for considered values of $r_i$ certain form of assemblage given by $\left\{Q_{b|0}\right\}_{b=0}^1,\left\{Q_{b|1}\right\}_{b=0}^1$ different up to relabeling is not on the edge, then the same is true when we omit assumption of difference of measurements (for assemblage of this certain form).

We will show that if $r_1=3$ and $r_3=3$ or $r_3=2$, then considered assemblages coming from measurements on rank three states are not on the edge. Observe that by performing measurements $\left\{Q_{b|0}\right\}_{b=0}^1,\left\{Q_{b|1}\right\}_{b=0}^1$ on $\rho_{a|x}$ we obtain some particular row from assemblage $\Sigma$
\begin{equation}\label{sub}
\begin{matrix}
\begin{array}{cc|cc}
   \sigma_{a0|x0}& \sigma_{a1|x0}&   \sigma_{a0|x1} &   \sigma_{a1|x1}.
\end{array}
\end{matrix}
\end{equation} Since $\rho_{a|x}\in M_2(\mathbb{C})\otimes M_2(\mathbb{C})$ and $\mathrm{rank}(\rho_{a|x})=3$ or $\mathrm{rank}(\rho_{a|x})=2$, we may use Lemma \ref{lem1rank3} and Lemma \ref{lem2rank2} in order to state what are possible ranks of operators in a considered row (\ref{sub}). From now on $\psi_i\in M_2(\mathbb{C})$ denotes some operator with rank one proportional to state $|\psi_i\rangle \langle \psi_i|$, 
$\psi'_i\in M_2(\mathbb{C})$ denotes either $\psi_i$ or $0$ and $X\in M_2(\mathbb{C})$ denotes any operator with full rank (i.e. rank two). Note that only possible structures of rows which are relevant for certification that a given assemblage is not on the edge, are those which have the largest number of operators with rank zero and rank one (i.e. if one can subtract an LHS part from assemblage with row of a certain structure, one can also subtract an LHS part from assemblages with row which arise from the previous structure by changing operator on some position to operator with rank two). 

Firstly, assume that $\left\{Q_{b|0}\right\}_{b=0}^1,\left\{Q_{b|1}\right\}_{b=0}^1$ are different (up to relabeling). According to Lemma \ref{lem1rank3}, when $\mathrm{rank}(\rho_{a|x})=3$, only relevant (up to relabeling $y\rightarrow y\oplus 1$ and $b_1\rightarrow b_1\oplus 1,b_2\rightarrow b_2\oplus 1$) rows are given as
\begin{equation}\label{row1}
\begin{matrix}
\begin{array}{cc|cc}
    \psi_1&   X&  \psi_2 & X.
\end{array}
\end{matrix}
\end{equation}
In the analogous way, according to Lemma \ref{lem2rank2}, when $\mathrm{rank}(\rho_{a|x})=2$, only relevant (up to relabeling $y\rightarrow y\oplus 1$ and $b_1\rightarrow b_1\oplus 1,b_2\rightarrow b_2\oplus 1$) rows are given as

\begin{equation}\label{row2}
\begin{matrix}
\begin{array}{cc|cc}
   0&   X&  X & X,
\end{array}
\end{matrix}
\end{equation}
\begin{equation}\label{row3}
\begin{matrix}
\begin{array}{cc|cc}
   \psi_1&   \psi_2&  X & X.
\end{array}
\end{matrix}
\end{equation}
\begin{equation}\label{row4}
\begin{matrix}
\begin{array}{cc|cc}
   \psi&   \psi&  \psi & \psi.
\end{array}
\end{matrix}
\end{equation}
\begin{equation}\label{row5}
\begin{matrix}
\begin{array}{cc|cc}
 \psi_1& X&  \psi_2 & X.
\end{array}
\end{matrix}
\end{equation}

As $r_1=3$ and $r_3=2,3$, we should analyze all assemblages with the first row and the third row described by (\ref{row1}) or (\ref{row2}-\ref{row5}) respectively (up to relabeling). For example by taking (\ref{row1}) with relabeling and (\ref{row4}) we get
\begin{equation}
\begin{pmatrix}
\begin{array}{cc|cc}
\color{red}X\color{black}&   \psi_1&   \psi_2&   \color{red}X\color{black}\\
             \ldots& \ldots&\ldots&\ldots \\ \hline
               \color{red}\psi\color{black}&   \psi&  \psi& \color{red}\psi\color{black}\\
							\ldots& \ldots&\ldots&\ldots
\end{array}
\end{pmatrix}
\end{equation}and one can see that there exists a choice of rectangle such that it does not contain any rank zero operator or any pair of rank one operators which are not proportional, i.e. such assemblage is not on the edge. Analysis of all possibilities like the one above (with relabeling) shows that there is only one problematic case given by
\begin{equation}\label{prob1}
\begin{pmatrix}
\begin{array}{cc|cc}
\psi_1&  X&   \psi_2&   X\\
\ldots& \ldots&\ldots&\ldots \\ \hline
X& \psi_3&X& \psi_4\\
\ldots& \ldots&\ldots&\ldots
\end{array}
\end{pmatrix}.
\end{equation}When (\ref{prob1}) is obtained by the same (up to relabeling) pair of measurements on the subsystem B, then $\psi_1=\psi_2$ and there is an LHS part that one can subtract from (\ref{prob1}), therefore we may still assume that $\left\{Q_{b|0}\right\}_{b=0}^1,\left\{Q_{b|1}\right\}_{b=0}^1$ are different (up to relabeling). Note that with (\ref{prob1}), neither $r_2$ nor $r_4$ can be equal to $0$ as in this case there would be a column in (\ref{prob1}) with operators with rank two, rank one and rank zero, which would be in contradiction with a possible form of this column given by (\ref{row1}-\ref{row5}) (as in that case $\left\{P_{a|0}\right\}_{a=0}^1,\left\{P_{a|1}\right\}_{a=0}^1$ should be different up to relabeling and column must be relater to some $q_i=2,3$). In similar fashion, if assemblage is of the form (\ref{prob1}), and $r_2$ or $r_4$ is equal to $1$, then considered assemblage is not on the edge. Indeed, it is enough to consider $r_2=1$. In that case once more $\left\{P_{a|0}\right\}_{a=0}^1,\left\{P_{a|1}\right\}_{a=0}^1$ are different (up to relabeling) an since $q_i=2,3$ for each $i=1,2,3,4$, (\ref{row1}-\ref{row5}) imply the following structure 
\begin{equation}
\begin{pmatrix}
\begin{array}{cc|cc}
\color{red}\psi_1\color{black}&  X&   \psi_2&   \color{red}X\color{black}\\
\phi_1& \phi_2&\phi_3&\phi_4 \\ \hline
X& \psi_3&X& \psi_4\\
\color{red}X\color{black}& X&X&\color{red}X\color{black}
\end{array}
\end{pmatrix}
\end{equation}from which one can certainly subtract an LHS part.

As a results of this considerations we see that analysis of (\ref{prob1}) boils down to $r_2,r_4=2,3$. The only situation problematic to subtraction of an LHS part from an assemblage is related to the case when the second and the fourth rows are given by $(\ref{row5})$ and assemblage (\ref{prob1}) is of the form given by (\ref{form2}) - this structure is however impossible if $\mathrm{rank}(\rho_{ABC})=3$ according to Lemma \ref{lem5rank3} (with different up to relabeling measurements on the subsystem B). This concludes the first part of the proof.

Due to the whole above analysis, it is enough to discuss only the case when $(r_1,r_2), (r_3,r_4)$ are of the form $(2,1),(2,2)$. Without the loss of generality we may put $r_1, r_3=2$. Analysis of (\ref{row2}-\ref{row5}) shows all possible obstructions for subtraction of an LHS part. The first potential obstruction is related to the case where both the first and the third row are like (\ref{row5}) and assemblage is given by (\ref{prob1}) - this obstruction is excluded due to the previous discussion. Another potential obstruction is related to the case where both the first and the third row have zero operators but in different columns. If both measurements on the subsystem A are the same up to relabeling one can subtract an LHS part from assemblage with mentioned property, so the only relevant case is without loss of generality given by the following structure (we use here form (\ref{row2}) for the first row and the third row) 
\begin{equation}\label{prob3}
\begin{pmatrix}
\begin{array}{cc|cc}
0&  X&   X&   X\\
X& X&\ldots&\ldots \\ \hline
X& 0&X& X\\
X& X&\ldots&\ldots
\end{array}
\end{pmatrix}
\end{equation}where at least two unspecified positions are occupied by a nonzero operator and an LHS part can be subtracted as well.

The remaining situation is related to the case where the first and the third rows are like (\ref{row3}-\ref{row4}). Take a convention in which 
\begin{equation}\nonumber
\begin{matrix}
\begin{array}{cc|cc}
   \sigma_1/\tilde{\sigma}_1&   \sigma_2/\tilde{\sigma}_2&  \sigma_3/\tilde{\sigma}_3 &  \sigma_4/\tilde{\sigma}_4,
\end{array}
\end{matrix}
\end{equation}stands for one of the following rows
\begin{equation}\nonumber
\begin{matrix}
\begin{array}{cc|cc}
   \sigma_1&   \sigma_2&  \sigma_3&  \sigma_4
\end{array}
\end{matrix},\ \begin{matrix}
\begin{array}{cc|cc}
 \tilde{\sigma}_1&   \tilde{\sigma}_2&  \tilde{\sigma}_3 & \tilde{\sigma}_4.
\end{array}
\end{matrix}
\end{equation}Firstly put $r_4=1$. There are two relevant cases to consider 
\begin{equation}\label{ex_1}
\begin{pmatrix}
\begin{array}{cc|cc}
\psi_1/\psi_I&   \psi_2/\psi_I&   X/\psi_I&   X/\psi_I\\
\ldots& \ldots&\ldots&\ldots \\ \hline
\psi_3/\psi_{II}&   \psi_4/\psi_{II}& X/\psi_{II} & X/\psi_{II}\\
\phi'_1& \phi'_2&\phi'_3&\phi'_4
\end{array}
\end{pmatrix}
\end{equation} and
\begin{equation}\label{ex_2}
\begin{pmatrix}
\begin{array}{cc|cc}
\psi_1/\psi_I&   \psi_2/\psi_I&   X/\psi_I&   X/\psi_I\\
\ldots& \ldots&\ldots&\ldots \\ \hline
X & X&\psi_3&   \psi_4\\
\phi'_1& \phi'_2&\phi'_3&\phi'_4
\end{array}
\end{pmatrix}
\end{equation}where $\phi'_i$ stands for operator proportional to state $|\phi_i\rangle \langle \phi_i|$ or operator equal to $0$. If assemblages (\ref{ex_1}-\ref{ex_2}) are obtained with $\left\{P_{a|0}\right\}_{a=0}^1,\left\{P_{a|1}\right\}_{a=0}^1$ which are the same (up to relabeling), one can subtract an LHS part. In other case form of the first two columns in (\ref{ex_1}) implies that $q_1=q_2=1$ which is in contradiction with Lemma \ref{lem1rank3} (as $\mathrm{rank}(\rho_{ABC})=3$). Similar contradiction is obtained for (\ref{ex_2}) if the first row consists of rank one operators. Due to (\ref{row1}-\ref{row5}) for columns only relevant case for (\ref{ex_2}) is given by 
\begin{equation}
\begin{pmatrix}
\begin{array}{cc|cc}
       \psi_1&   \psi_2&   X&   X\\
\color{red}X\color{black}& X&X&\color{red}X\color{black} \\ \hline
\color{red}X\color{black} & X&\psi_3&   \color{red}\psi_4\color{black}\\
\phi'_1& \phi'_2&\phi'_3&\phi'_4
\end{array}
\end{pmatrix}
\end{equation}which cannot be on the edge.

To consider the remaining case for $r_1=r_3=2$, put $r_4=2$. Without loss of generality, we may consider three particular examples for such situation
\begin{equation}\label{exx_1}
\begin{pmatrix}
\begin{array}{cc|cc}
\psi_1/\psi_{I}&   \psi_2/\psi_{I}&   X/\psi_{I}&   X/\psi_{I}\\
\ldots& \ldots&\ldots&\ldots \\ \hline
\psi_3/\psi_{II}& \psi_4/\psi_{II}&X/\psi_{II}&X/\psi_{II}\\
\psi_5/\psi_{III}&   \psi_6/\psi_{III}& X/\psi_{III} & X/\psi_{III}
\end{array}
\end{pmatrix},
\end{equation}
\begin{equation}\label{exx_2}
\begin{pmatrix}
\begin{array}{cc|cc}
\psi_1/\psi_I&   \psi_2/\psi_I&   X/\psi_I&   X/\psi_I\\
\ldots& \ldots&\ldots&\ldots \\ \hline
X&X&\psi_3& \psi_4\\
\psi_5/\psi_{III}&   \psi_6/\psi_{III}& X/\psi_{III} & X/\psi_{III}
\end{array}
\end{pmatrix}
\end{equation}and 
\begin{equation}\label{exx_3}
\begin{pmatrix}
\begin{array}{cc|cc}
\psi_1&   \psi_2&   X&   X\\
\ldots& \ldots&\ldots&\ldots \\ \hline
X&X&\psi_3& \psi_4\\
X & X&\psi_5&   \psi_6
\end{array}
\end{pmatrix}.
\end{equation}Once more, if assemblages (\ref{exx_1}-\ref{exx_2}) are obtained with $\left\{P_{a|0}\right\}_{a=0}^1,\left\{P_{a|1}\right\}_{a=0}^1$ which are the same (up to relabeling), one can subtract an LHS part - assume that it is not the case. Similar to the previous discussion the first two columns in (\ref{exx_1}) imply that $q_1=q_2=1$ which is in contradiction with Lemma \ref{lem1rank3}. The same is true for (\ref{exx_2}) if the first row and the fourth row consists of rank one operators. If the first row in (\ref{exx_2}) consists only of rank one operators and the fourth row have rank two operators, then the second row consists of rank two operators and one can subtract an LHS part from this structure (the same is true is the role of the first and the fourth rows are reversed) - this is true since in that case $q_3,q_4=2,3$ and related columns must be given by (\ref{row2}-\ref{row5}).

In the end, if both the first and the fourth row have rank two operators, only obstruction for not being on the edge would be given by 
\begin{equation}\label{exxx_2}
\begin{pmatrix}
\begin{array}{cc|cc}
\psi_1&   \psi_2&   X&   X\\
X& X&\phi_1&\phi_2 \\ \hline
X&X&\psi_3& \psi_4\\
\psi_5&   \psi_6& X & X
\end{array}
\end{pmatrix}.
\end{equation}However, note that (\ref{exxx_2}) is the same as (\ref{form2}) up to exchange of measurements on subsystems A and B (and relabeling). Therefore, by Lemma \ref{lem5rank3}, this case can be neglected.

Finally, by analogous arguments, one can see that the only structure of type (\ref{exx_3}) which may provide an obstacle for the subtraction of an LHS part is given by formula (\ref{form}), which once more according to Lemma \ref{lem5rank3} is irrelevant for $\rho_{ABC}$ with rank three. The proof is now completed in the case of nontrivial projective measurement.

Now let us consider the case of trivial projective measurements given by $0, \mathds{1}$. Firstly, assume that $P_{0|0}=Q_{0|0}=0$. Then $\Sigma$ has the first column and the first row which consists of zeros and due to no-signaling and normalization conditions one can see that there always exists an LHS part that can be subtracted. Therefore, without loss of generality if there exist at least one trivial measurement on both subsystems, obtained assemblage is not on the edge. If so, assume that measurement on the subsystem B are given by nontrivial projective measurements (we may put them to be different up to relabeling). Now put $P_{0|0}=P_{0|1}=0$. Clearly assemblage obtained by this measurements have the first row and the third row which consists of zeros, while the remaining rows are the same and they consist of operators obtained after measurements on the subsystem B performed on reduced state $\rho_{BC}=\mathrm{Tr}_A(\rho_{ABC})$. Note that $\mathrm{rank}(\rho_{BC})\geq 2$ since $\rho_{BC}=\rho_{0|0}+\rho_{1|0}$ and $\mathrm{max}\left\{r_1,r_2\right\}\geq 2$ by previous discussion (here both $\rho_{a|x}$ are assumed to be obtained by nontrivial projective measurement). Without loss of generality we may then assume that $\mathrm{rank}(\rho_{BC})=2$. In this case the relevant row is given by one of $(\ref{row2}-\ref{row5})$ and it is obvious that one can subtract an LHS part from the considered assemblage. To conclude this part of the proof, consider the case in which $P_{0|0}=\mathds{1}$ and $\left\{P_{a|1}\right\}_{a=0}^1$ is nontrivial. We may discuss only relevant case for which $r_3=2$ and the third row is given by one of $(\ref{row2}-\ref{row5})$. By no-signaling constraints for columns we see that assemblage obtained with such measurements is not on the edge, for example 
\begin{equation}\nonumber
\begin{pmatrix}
\begin{array}{cc|cc}
 \color{red}\psi_{1}\color{black}+\sigma_{10|10}&X+\sigma_{11|10}&\psi_2+\sigma_{10|11}&\color{red}X\color{black}+\sigma_{11|11}\\
0& 0&0&0\\ \hline
\color{red}\psi_{1}\color{black}&X&\psi_2&\color{red}X\color{black}\\
\sigma_{10|10} & \sigma_{11|10}&\sigma_{10|11}& \sigma_{11|11}
\end{array}
\end{pmatrix}.
\end{equation}Therefore, any no-signaling assemblage obtained from projective measurements performed on rank three state is not on the edge.

Finally, we can extend the proof in order to cover the case of all POVMs. Consider once more an assemblage $\sigma_{ab|xy}=\mathrm{Tr}_{AB}(M_{a|x}\otimes N_{b|y} \otimes \mathds{1}\rho_{ABC})$ where $M_{a|x}$, $N_{b|y}$ are elements of nontrivial POVMs acting on the subsystem A and B respectively. In this case 
\begin{align*}
   M_{0|0}&= m_{0|0}P_{0|0} + m_{1|0}P_{1|0}\\ 
   M_{1|0}&= (1-m_{0|0})P_{0|0} + (1-m_{1|0})P_{1|0}\\
   M_{0|1}&= m_{0|1}P_{0|1}+m_{1|1}P_{1|1}\\
   M_{1|1}&= (1-m_{0|1})P_{0|1} + (1-m_{1|1})P_{1|1}
\end{align*}
where $0<m_{0|0}, m_{1|0}, m_{0|1}, m_{1|1} <1$ and $P_{a|x}$ are orthogonal projections coming from spectral decomposition of appropriate $M_{a|x}$. Define $\alpha= \textnormal{min} \{m_{0|0},1-m_{1|0},m_{0|1},1-m_{1|1} \}>0$. Using previous equalities we can write 
$M_{a|x}=\alpha P_{a|x} + \tilde{M}_{a|x}$ for all $a,x=0,1$ where each $\tilde{M}_{a|x}$ is a positive operator. By analogous reasoning we may write $N_{b|y} = \beta Q_{b|y} + \tilde{N}_{b|y}$ for all $b,y =0,1$ and some $\beta > 0$. Inserting the above expressions for the POVMs into the definition of $\sigma_{ab|xy}$, we get $\sigma_{ab|xy}=\alpha \beta \tilde{\sigma}_{ab|xy} + \sigma_{ab|xy}^{\textnormal{junk}}$ where $\tilde{\sigma}_{ab|xy} = \mathrm{Tr}_{AB}(P_{a|x}\otimes Q_{b|y} \otimes \mathds{1}\rho_{ABC})$ is an assemblage obtained by nontrivial projective measurements realized on a tripartite state $\rho_{ABC}$ with $\mathrm{rank} (\rho_{ABC}) \geq 3$ and $\sigma_{ab|xy}^{\textnormal{junk}}$ is some irrelevant no-signaling assemblage (up to normalization equal to $1- \alpha \beta$). We know that $\tilde{\sigma}_{ab|xy}$ is not on the edge (we can subtract an LHS part from it) and therefore the initial assemblage $\sigma_{ab|xy}$ is neither on the edge. Using similar reasoning in the case with some trivial POVMs, we obtain the same outcome based on previous results. Therefore, the proof of Theorem \ref{thm_main} is concluded.
\end{proof}

\end{document}